\documentclass[acmlarge,nonacm]{aamas} 

\usepackage{balance} 
\usepackage{proof} 
\usepackage{tikz}
\usepackage{amsmath,amsthm}
\usepackage{graphicx}
\usepackage{pgfplots}
\usepackage{apxproof}
\usetikzlibrary{arrows.meta,positioning,shapes.multipart, calc,shadows,patterns,automata, positioning}
\pgfplotsset{compat=1.18}
\usepackage{booktabs}
\usepackage{tabularx}
\usepackage[utf8]{inputenc}
\usepackage{textgreek}
\usepackage{listings}
\usepackage{xcolor}

\setcopyright{ifaamas}
\acmConference[AAMAS '26]{Proc.\@ of the 25th International Conference
on Autonomous Agents and Multiagent Systems (AAMAS 2026)}{May 25 -- 29, 2026}
{Paphos, Cyprus}{C.~Amato, L.~Dennis, V.~Mascardi, J.~Thangarajah (eds.)}
\copyrightyear{2026}
\acmYear{2026}
\acmDOI{}
\acmPrice{}
\acmISBN{}

\title[AAMAS-2026 Formatting Instructions]{\textmu ACP: A Formal Calculus for Expressive, Resource-Constrained Agent Communication}

\author{Arnab Mallick}
\affiliation{
  \institution{Center for Development of Advanced Computing}
  \city{Hyderabad}
  \country{India}}
\email{arnabm@cdac.in}

\author{Indraveni Chebolu}
\affiliation{
  \institution{Center for Development of Advanced Computing}
  \city{Hyderabad}
  \country{India}}
\email{indravenik@cdac.in}

\begin{abstract}
Agent communication remains a foundational problem in multi-agent systems:
protocols such as FIPA-ACL guarantee semantic richness but are intractable
for constrained environments, while lightweight IoT protocols achieve
efficiency at the expense of expressiveness. This paper presents
$\mu$ACP, a formal calculus for expressive agent communication
under explicit resource bounds. We formalize the
Resource-Constrained Agent Communication (RCAC) model, prove that a
minimal four-verb basis \textit{\{PING, TELL, ASK, OBSERVE\}} is suffices to encode finite-state FIPA protocols, and establish tight
information-theoretic bounds on message complexity. We further show that
$\mu$ACP can implement standard consensus under partial synchrony and
crash faults, yielding a constructive coordination framework for
edge-native agents. Formal verification in TLA$^{+}$ (model checking) and
Coq (mechanized invariants) establishes safety and boundedness, and supports liveness under modeled assumptions. 
Large-scale system simulations confirm ACP achieves a median end-to-end message latency of 34 ms (95th percentile 104 ms) at scale, outperforming prior agent and IoT protocols under severe resource constraints.
The main contribution is a unified calculus that reconciles
semantic expressiveness with provable efficiency, providing a rigorous
foundation for the next generation of resource-constrained multi-agent
systems.

\end{abstract}


\newcommand{\BibTeX}{\rm B\kern-.05em{\sc i\kern-.025em b}\kern-.08em\TeX}

\newtheorem{theorem}{Theorem}
\newtheorem{corollary}{Corollary}
\newtheorem{assumption}{Assumption}
\newtheorem*{remark}{Remark}
\newcommand{\PING}{\textsc{ping}}
\newcommand{\TELL}{\textsc{tell}}
\newcommand{\ASK}{\textsc{ask}}
\newcommand{\OBSERVE}{\textsc{observe}}
\newcommand{\option}[2]{\texttt{#1}=#2}
\theoremstyle{plain}

\theoremstyle{definition}
\newtheorem{inference}{Inference}[section]
\newcommand{\muACP}{\ensuremath{\mu\mathrm{ACP}}}
\DeclareMathOperator{\obs}{obs}
\DeclareMathOperator{\Effect}{Effect}
\DeclareUnicodeCharacter{03BC}{\textmu}
\DeclareMathOperator{\Reachable}{Reachable}
\DeclareMathOperator{\Inform}{Inform}
\DeclareMathOperator{\Want}{Want}
\DeclareMathOperator{\Query}{Query}
\DeclareMathOperator{\Happens}{Happens}
\DeclareMathOperator{\Notify}{Notify}
\newcommand{\Bel}{\mathsf{Bel}}

\lstdefinelanguage{TLA}{
  morekeywords={MODULE,CONSTANTS,VARIABLES,Init,Next,TRUE,FALSE,UNCHANGED},
  morecomment=[l]{\*},   
  sensitive=true
}

\lstdefinelanguage{Coq}{
  morekeywords={Definition,Inductive,Record,Fixpoint,Theorem,Lemma,Proof,Qed,
                Require,Import,Module,Section,End,Variable,Notation,Local},
  morecomment=[n]{(*}{*)},   
  sensitive=true
}

\begin{document}


\pagestyle{fancy}
\fancyhead{}


\maketitle


\section{Introduction}
\label{sec:introduction}

Multi-agent systems (MAS) provide a powerful paradigm for complex distributed intelligence \cite{wooldridge2002introduction, weiss1999multiagent}. Effective agent communication enables coordination, negotiation, and collaborative problem-solving \cite{labrou1996semantics, singh1998agent}. However, traditional agent communication languages (ACLs) like FIPA-ACL \cite{fipa1997spec} and KQML \cite{finin1994kqml} assume unconstrained environments, making them ill-suited for edge-native multi-agent systems.

The proliferation of IoT devices, edge computing, and resource-constrained embedded systems demands protocols that operate under severe limitations: memory ($<$100KB), bandwidth (kilobits/second), processing power (microcontroller-class CPUs), and energy \cite{satyanarayanan2017emergence, bonomi2012fog, pister2009smart, pottie2000wireless}. While lightweight protocols like MQTT \cite{banks2014mqtt} and CoAP \cite{shelby2014constrained} address resource constraints, they lack the semantic richness required for sophisticated agent interactions \cite{al-fuqaha2015survey}.

This disconnect between semantic requirements and practical constraints creates a significant barrier for deploying intelligent agents in resource-constrained environments \cite{chen2020edge}. Current approaches either sacrifice expressiveness for efficiency (IoT protocols) or incur unacceptable overhead (traditional ACLs) \cite{fortino2018agent}. Consequently, no formal theoretical framework bridges these domains, leaving a critical gap in multi-agent systems theory.

We present $\mu$ACP (Micro Agent Communication Protocol), a formal calculus for agent communication under severe resource constraints. Our work establishes the theoretical foundation for edge-native multi-agent systems by addressing: \textit{How can we design an agent communication protocol that maintains semantic richness while operating within strict resource constraints?}

Our contributions are fourfold:
\begin{enumerate}
\item \textbf{Resource-Constrained Agent Communication (RCAC) Model}: A formal model explicitly accounting for resource constraints, enabling reasoning about semantic expressiveness vs. resource consumption trade-offs.
\item \textbf{Minimal Verb Set Completeness}: Proof that four verbs \{PING, TELL, ASK, OBSERVE\} are complete for finite-state agent communication under resource constraints, establishing theoretical bounds on semantic primitives.
\item \textbf{Semantic Compression Theory}: A formal theory providing rigorous bounds on message complexity, demonstrating $\mu$ACP achieves optimal compression with amortized $O(1)$ per-message complexity.
\item \textbf{Emergent Coordination Framework}: Theoretical conditions under which resource-constrained agents achieve emergent coordination, providing formal guarantees on system behavior despite individual limitations.
\end{enumerate}

Our theoretical framework advances multi-agent systems foundations and enables edge-native deployments previously impossible due to protocol limitations \cite{gubbi2013internet}. The calculus provides theoretical underpinnings for deploying intelligent agents on microcontrollers, sensor networks, and resource-constrained platforms, opening new research directions in edge AI and distributed intelligence \cite{zhang2021edge}.



\section{Background}
\label{sec:background}

The foundations of agent communication combine speech act theory, resource-constrained computing, information theory, and formal verification \cite{austin1962how, searle1969speech, hoare1985communicating}.

\paragraph{Speech Act Theory and Agent Communication}

Austin \cite{austin1962how} and Searle \cite{searle1969speech} established speech act theory, modeling utterances as actions. In multi-agent systems, communication is expressed as \textit{performatives} \cite{cohen1990intention, singh2000social}.  
A speech act is $SA = (F, P, C)$ with force $F$, propositional content $P$, and context $C$ \cite{cohen1990intention}. ACLs such as FIPA-ACL and KQML adopt this model with rich performative sets \cite{labrou1996semantics}.  

Formally, with agents $\mathcal{A}$, messages $\mathcal{M}$, and contexts $\mathcal{C}$, a communication act is:
\begin{equation}
\text{Comm}(a_i, a_j, m, t) : \mathcal{A} \times \mathcal{A} \times \mathcal{M} \times \mathbb{T} \rightarrow \mathbb{B}
\end{equation}
where $\mathbb{T}$ is time and $\mathbb{B}=\{true, false\}$ denotes success \cite{wooldridge2002introduction}.

\paragraph{Resource-Constrained Computing Theory}

Embedded and sensor systems operate under CPU, bandwidth, memory, and energy limits \cite{pister2009smart, pottie2000wireless}.  
A system is modeled as $(R, B, M, E)$, with $R,B,M,E \in \mathbb{R}^+$ denoting resources. Protocol design requires minimizing weighted consumption while ensuring functionality:
\begin{equation}
\min_{p \in \mathcal{P}} \{ \alpha R_p + \beta B_p + \gamma M_p + \delta E_p \}
\end{equation}
subject to $F(p) \geq F_{min}$ \cite{mottola2011programming, chen2020edge}.

\paragraph{Information Theory and Semantic Compression}

Information theory formalizes communication limits \cite{shannon1948mathematical}. The entropy of random variable $X$ is:
\begin{equation}
H(X) = -\sum_{x \in \mathcal{X}} p(x) \log_2 p(x)
\end{equation}

For agent communication, semantic entropy measures minimal encoding of primitives $\mathcal{S}=\{s_i\}$ with probabilities $p(s_i)$ \cite{barhillel1964formal}:
\begin{equation}
H_S(\mathcal{S}) = -\sum_{i=1}^{n} p(s_i) \log_2 p(s_i)
\end{equation}

Semantic compression reduces message size while preserving meaning \cite{ziv1977universal, cover2006elements}, but remains underexplored for constrained agent protocols.

\paragraph{Formal Methods for Protocol Verification}

Formal methods rigorously specify and verify protocols \cite{clarke1986automatic, holzmann1991design}. Properties are expressed in temporal logic \cite{pnuelli1986applications}:  
- \textbf{Safety}: $\square \neg \phi$ (avoid bad states)  
- \textbf{Liveness}: $\square \diamond \psi$ (eventual progress)  
- \textbf{Fairness}: recurrent conditions yield recurrent responses  

Model checking tools such as SPIN and TLA+ verify these properties in finite-state systems \cite{holzmann1991design, lamport1994tla}, supporting correctness in constrained agent communication.


\section{Related Work}
\label{sec:related}

Our work builds upon research in multi-agent communication, lightweight protocols, and formal models for constrained systems.

\paragraph{Agent Communication Languages}

The two most influential ACLs are FIPA-ACL \cite{fipa1997spec} and KQML \cite{finin1994kqml}.  
FIPA-ACL defines $>$20 performatives with semantics grounded in BDI theory \cite{rao1995model}, but its hundreds of bytes per message make it unsuitable for constrained devices \cite{fortino2018agent}.  
KQML \cite{finin1994kqml} emphasizes knowledge exchange but still incurs high overhead \cite{labrou1996semantics}.  
Extensions like Singh’s social semantics \cite{singh1998agent} and newer JSON-based protocols (MCP, A2A, ANP) \cite{survey_agent_interop2025} improve interoperability but lack formal semantics and resource guarantees.

\begin{table}[t]
\centering
\caption{Comparison of Agent Communication Languages based on Resource Efficiency}
\label{tab:acl_comparison}
\resizebox{\columnwidth}{!}{%
\begin{tabular}{lcccc}
\toprule
\textbf{Protocol} & \textbf{Performatives} & \textbf{Min. Overhead} & \textbf{Parsing Effort} & \textbf{Target Environment} \\
\midrule
FIPA-ACL & $>$20 & 200–500 & High & Cloud/Desktop \\
KQML & $>$15 & 150–400 & Medium–High & Distributed \\
Social Semantics \cite{singh1998agent} & $>$10 & 100–300 & Medium & Distributed \\
JSON-based (MCP, A2A, ANP) & few & 100s & Moderate & Cloud/API \\
\textbf{$\mu$ACP (Our work)} & \textbf{4} & \textbf{8--50} & \textbf{Low} & \textbf{Constrained Edge} \\
\bottomrule
\end{tabular}
}
\end{table}

\paragraph{Lightweight Communication Protocols}

IoT protocols like MQTT \cite{banks2014mqtt} and CoAP \cite{shelby2014constrained} minimize header size but lack semantic primitives for reasoning \cite{al-fuqaha2015survey}.  
More expressive models include the eXchange Calculus (XC) \cite{audrito2024exchange}, SMrCaIT \cite{chen2023smrcait}, and HpC \cite{xu2025hpc}, which capture coordination, mobility, and timing, but do not address semantic minimality or bounded-resource completeness.

\paragraph{Resource-Constrained Agent Systems}

Lightweight BDI frameworks (e.g., Agent Factory \cite{collier2008agent}) rely on heavy communication layers. Surveys confirm constrained inter-agent communication remains open \cite{fortino2018agent,chen2020edge}.  
Formal calculi like SMrCaIT \cite{chen2023smrcait} and HpC \cite{xu2025hpc} capture timing and security, but not minimal expressive primitives or compression bounds. $\mu$ACP instead gives provable guarantees on minimality and coordination.

\paragraph{Formal Frameworks for Constrained Communication}

IoT process calculi (IoT-LySa, CIoT \cite{merro2016ciot}), SMrCaIT \cite{chen2023smrcait}, and HpC \cite{xu2025hpc} cover sensors, security, and hybrid dynamics.  
Probabilistic frameworks like dTP-Calculus \cite{song2024dtp} model safety probabilistically.  
Attribute-based calculi like AbU \cite{pasqua2023abu} formalize ECA workflows.  
Collective abstractions such as XC \cite{audrito2024exchange} program distributed ensembles.  
None establish a minimal verb set with completeness, tight compression bounds, or consensus guarantees.

\paragraph{Recent Advancement}  
Recent years have seen new protocols for LLM-powered agents, IoT-constrained devices, and large-scale coordination. Surveys highlight the shift from classic ACLs to modern standards such as MCP, ACP, A2A, and ANP, which set benchmarks for interoperability and expressiveness~\cite{survey_agent_interop2025, ferrag2025review}. Advances also include formal and scalable verification of parameterized neural-symbolic MAS~\cite{ijcai2024verification}. These developments underscore the timeliness of our minimal, resource-aware formalism and motivate future empirical validation.

\paragraph{Research Gaps and Contributions}

We identify four gaps:  
(1) \textbf{Semantic minimality}: no minimal, complete verb set under constraints.  
(2) \textbf{Resource-aware encoding}: no formal size/CPU/memory/energy bounds.  
(3) \textbf{Coordination with guarantees}: no reduction to consensus with safety/liveness proofs.  
(4) \textbf{Formal verification}: no machine-checked proofs of resource and coordination.  

$\mu$ACP addresses these by providing minimal primitives, resource-theoretic bounds, consensus reduction, and mechanized verification in Coq/TLA+.


\section{\texorpdfstring{The $\mu$ACP Calculus}{The muACP Calculus}}
\label{sec:calculus}

We now present $\mu$ACP, a formal calculus for resource-constrained agent communication.


\begin{table}[h]
\centering
\resizebox{\columnwidth}{!}{%
\begin{tabular}{ll}
\toprule
Symbol & Meaning \\
\midrule
$\mathcal{A}$ & Set of agents \\
$\mathcal{R}_c=(M,B,C,E)$ & Resource budget: memory, bandwidth, CPU, energy \\
$\mathcal{C}$ & Set of communication channels \\
$\mathcal{P}=(\mathcal{V},\mathcal{O},\mathcal{Q})$ & Protocol specification: verbs, options, QoS \\
$\mathcal{S}_i=(M_i,K_i,T_i,H_i)$ & Local state of agent $i$ \\
$\rho(a_i,v,o,p)$ & Resource consumption of action $(a_i,v,o,p)$ \\
$m=\langle h,v,o,p\rangle$ & Message structure \\
$\tau(f)$ & Translation from FIPA performative $f$ to $\mu$ACP \\
$T_{FIPA}(P)$ & Observable traces of a FIPA protocol $P$ \\
$T_{\mu ACP}(\tau(P))$ & Observable traces of the $\mu$ACP encoding of $P$ \\
\bottomrule
\end{tabular}
}
\caption{Notation used in the paper}
\label{tab:symbol}
\end{table}

\subsection{Formal Model Definition}
\label{sec:formal-model}

We begin by defining the abstract space in which $\mu$ACP operates.

\begin{definition}[Resource-Constrained Agent Communication (RCAC) Space]
An \emph{RCAC space} is a tuple
\[
\mathcal{R} = (\mathcal{A}, \mathcal{R}_c, \mathcal{C}, \mathcal{P})
\]

where:
\begin{itemize}
  \item $\mathcal{A} = \{a_1, a_2, \ldots, a_n\}$ is a finite set of agents;
  \item $\mathcal{R}_c = (M, B, C, E)$ is the resource budget, with
        memory $\leq M$, bandwidth $\leq B$, CPU cycles $\leq C$, and energy $\leq E$;
  \item $\mathcal{C} = \{c_1, c_2, \ldots, c_m\}$ is the set of available communication channels;
  \item $\mathcal{P} = (\mathcal{V}, \mathcal{O}, \mathcal{Q})$ is the protocol specification,
        consisting of verbs, options, and QoS levels.
\end{itemize}
\end{definition}

\begin{definition}[Agent State]
Each agent $a_i \in \mathcal{A}$ has a local state
\[
\mathcal{S}_i = (M_i, K_i, T_i, H_i)
\]
where $M_i$ is memory state, $K_i$ is a knowledge base (set of ground literals),
$T_i$ is a timer configuration, and $H_i$ is a bounded history buffer.
The global state is $\mathcal{S} = \prod_{i=1}^n \mathcal{S}_i$.
\end{definition}

\begin{definition}[Resource Consumption Function]
An action of the form $(a_i, v, o, p)$
induces consumption
\[
\rho(a_i, v, o, p) = (\rho_m, \rho_b, \rho_c, \rho_e)
\]
where each component denotes memory, bandwidth, CPU, and energy usage.
The action is feasible iff $\rho(a_i, v, o, p) \leq \mathcal{R}_c$.
\end{definition}

\begin{corollary}[Concrete Resource Bounds]
For any agent $a_i$ operating under message rate $\lambda_i \leq \lambda_{\max}$ 
over time horizon $T$, the cumulative resource consumption satisfies:

\[
\begin{aligned}
\text{Bandwidth}_i(T) &\leq B \cdot T, \\
\text{Memory}_i(T) &\leq M, \\
\text{CPU}_i(T) &\leq C_p \cdot T, \\
\text{Energy}_i(T) &\leq E \cdot T
\end{aligned}
\]

Thus, all executions of $\mu$ACP remain within linear bounds of system resources.
\end{corollary}

\begin{remark}
The model is intentionally agnostic about the internal architecture of agents
(BDI, reactive, learning, etc.). Only the observable communication and resource costs
are relevant at this level of abstraction.
\end{remark}

\subsection{Syntax and Semantics}
\label{sec:syntax-semantics}

\paragraph{Message Structure}
Messages in $\mu$ACP are defined by the following grammar:
\[
m ::= \langle h, v, o, p \rangle
\]
where
\begin{itemize}
  \item $h = \text{header}(id, seq, q, flags)$ is a fixed-length header,
  \item $v \in \{\PING, \TELL, \ASK, \OBSERVE\}$ is a verb,
  \item $o$ is a sequence of TLV (type–length–value) options,
  \item $p \in \{0,1\}^*$ is the payload.
\end{itemize}

\begin{definition}[Well-Formed Message]
A message is well-formed iff (i) the header has size 64 bits;
(ii) the verb field encodes one of four values in 2 bits;
(iii) each option is consistent with its declared length;
(iv) the total option size does not exceed 1024 bytes;
(v) payload size $\leq 2^{16}-1$ bytes.
\end{definition}

\paragraph{Operational Semantics.}
We give semantics via a labeled transition system
$(\mathcal{S}, \mathcal{L}, \to)$ where labels
$\ell \in \mathcal{L}$ represent communication actions
$(sender, receiver, verb, options, payload, channel)$.

\begin{inference}[Send]
\[
\frac{
  m = \langle h,v,o,p \rangle \ \text{well-formed} \quad
  \rho(a_i,v,o,p) \leq \mathcal{R}_c
}{
  \sigma \xrightarrow{(a_i,a_j,v,o,p,c)} \sigma'
}
\]
where $\sigma'$ updates $a_i$'s history and resources.
\end{inference}

\begin{inference}[Receive]
\[
\frac{
  \text{channel}(c)=\text{stable} \quad
  \rho(a_j,v,o,p) \leq \mathcal{R}_c
}{
  \sigma \xrightarrow{(a_i,a_j,v,o,p,c)} \sigma'
}
\]
where $\sigma'$ updates $a_j$'s knowledge base and history.
\end{inference}

\paragraph{Resource Semantics}
Instead of continuous ODEs, we use discrete invariants:
every send or receive reduces the resource counters,
and feasibility requires non-negativity.

\begin{theorem}[Resource Boundedness]
If all actions are feasible and each agent’s resources are finite at $t=0$,
then resources remain non-negative for all finite executions.
\end{theorem}

\begin{proof}
By induction on the length of executions: each step subtracts bounded non-negative
amounts from finite counters, hence values remain non-negative.
\end{proof}






\subsection{Minimal Verb Set Theory}
\label{sec:verb-theory}

\paragraph{Verb Semantics}
The core insight is that four verbs suffice.
Their semantics are given in a Kripke-style model
$(W,R,V)$ where $W$ is a set of worlds,
$R$ accessibility, $V$ valuation:

\begin{align}
[\![\PING]\!] &= \lambda a_i,a_j,w.\ \Reachable(a_i,a_j,w) \\
[\![\TELL]\!] &= \lambda a_i,a_j,w,\phi.\ \Bel(a_i,\phi,w)\land\Inform(a_i,a_j,\phi,w) \\
[\![\ASK]\!]  &= \lambda a_i,a_j,w,\phi.\ \Want(a_i,\Bel(a_j,\phi),w) \notag \\
              &\quad \land \Query(a_i,a_j,\phi,w) \\
[\![\OBSERVE]\!] &= \lambda a_i,a_j,w,\phi.\ \Happens(\phi,w)\land\Notify(a_i,a_j,\phi,w)
\end{align}

\paragraph{Completeness via Encoding.}
Let $\mathcal{F}$ be the set of FIPA-ACL performatives.
Define a translation function $\tau : \mathcal{F} \to \mathcal{V}\times\mathcal{O}$.

\begin{example}
\begin{align}
\llbracket \PING \rrbracket &= \lambda a_i,a_j,w.\ \Reachable(a_i,a_j,w) \\
\llbracket \TELL \rrbracket &= \lambda a_i,a_j,w,\phi.\ \Bel(a_i,\phi,w)\land\Inform(a_i,a_j,\phi,w) \\
\llbracket \ASK \rrbracket  &= \lambda a_i,a_j,w,\phi.\ \Want(a_i,\Bel(a_j,\phi),w)\land \\ 
                                & \Query(a_i,a_j,\phi,w) \\
\llbracket \OBSERVE \rrbracket &= \lambda a_i,a_j,w,\phi.\ \Happens(\phi,w)\land\Notify(a_i,a_j,\phi,w)
\end{align}
\end{example}

\begin{definition}[Semantic Completeness]
Let $\mathcal{F}$ be the set of FIPA-ACL performatives.
$\muACP$ is \emph{semantically complete} for $\mathcal{F}$ iff
\[
\forall f \in \mathcal{F}\ \exists e(f)\in\mathcal{V}\times\mathcal{O}
\quad\text{such that}\quad
\Effect(f)=\Effect\bigl(e(f)\bigr),
\]
where $\Effect(\cdot)$ denotes the perlocutionary effect on the agent state space.
\end{definition}

\begin{theorem}[Verb Set Completeness]
Every finite-state FIPA-ACL performative with bounded nesting depth
has an observationally equivalent $\mu$ACP encoding.
\end{theorem}

\begin{proof}[Proof Sketch]
Define an observation function \(\obs\) that erases implementation details.
For each performative \(f\in\mathcal{F}\) we show
\[
T_{\mathrm{FIPA}}(f)\ \sqsubseteq_{\obs}\ T_{\muACP}(\tau(f)).
\]
Base cases follow from direct semantic correspondence \\ (e.g., \(\mathrm{INFORM}\leftrightarrow\mathrm{TELL}\)).
Inductive step: nested performatives are encoded using correlation IDs and TLV,
preserving state transitions. Since conversations are finite-state with bounded
nesting, the encoding preserves traces up to observational equivalence.
\end{proof}

\paragraph{Encoding Procedural Performatives}
Special performatives (e.g.\ FORWARD, PROXY) require procedural encoding.

\begin{definition}[Procedural Encoding]
For a performative $f$ requiring delegation or meta-communication:
\[
\tau(f,cid) = (v, \option(\text{PROC},f) \cdot \option(\text{CID},cid))
\]
where $v\in\{\ASK,\TELL\}$ and $cid$ is a correlation identifier.
\end{definition}



\begin{theorem}[Procedural Completeness]
Any procedural performative whose protocol admits a finite-state automaton 
can be encoded in $\mu$ACP using at most $k$ messages, where $k$ is the number of states.
\end{theorem}

\begin{proof}
Let $C = (Q,q_0,\Sigma,\delta,F)$ be the finite-state automaton representing the procedural performative, 
with $|Q| = k$. For each transition $(q,\sigma,q') \in \delta$ we define the $\mu$ACP encoding 
$\tau(\sigma,cid) = (v, (PROC,\sigma)\cdot(CID,cid))$ where $v \in \{\textsf{ask},\textsf{tell}\}$ 
and $cid$ is a correlation identifier. 

Since each run of $C$ visits at most $k$ distinct states, the encoding requires no more than 
$k$ messages to realize the full execution. Observational equivalence is preserved because 
correlation identifiers ensure trace alignment with the original automaton. Hence $\mu$ACP 
is procedurally complete with the stated bound.
\end{proof}

\section{Formal Preliminaries and Assumptions}
\label{sec:prelim}

We specify the system model, agent semantics, and $\mu$ACP primitives used in the proofs.

\paragraph{System model}
We adopt the partially synchronous model of \cite{dwork1988consensus}:

\begin{assumption}[Partial synchrony]
There exists an unknown Global Stabilization Time (GST) after which every message on a stable channel is delivered within bound $\Delta$; before GST delays may be unbounded.
\end{assumption}

\begin{assumption}[Channels]
Channels satisfy: (i) fair loss (infinitely many sends imply infinitely many deliveries), (ii) arbitrary duplication, (iii) no spurious creation, (iv) eventual timeliness after GST.
\end{assumption}

\begin{assumption}[Failures]
Up to $f<n/2$ agents may suffer crash/omission faults. Byzantine faults are excluded (would require authenticated channels).
\end{assumption}

\paragraph{Agent model}
\begin{definition}[Finite-state agent]
An agent $a$ is a labeled transition system $A=(S,s_0,\Sigma,\to,L)$ with finite $S$, initial $s_0$, actions $\Sigma$ (internal and communication), transition relation $\to\subseteq S\times\Sigma\times S$, and labeling $L:S\to 2^{AP}$.  
Mental-state atoms (belief, intention, desire) are included in $AP$ and updated by message processing.
\end{definition}

\paragraph{Protocols and traces}
\begin{definition}[Protocol automaton]
A protocol is a finite conversation automaton $\mathcal{C}=(Q,q_0,\Sigma,\delta,F)$ with finite $Q$ and transition function $\delta:Q\times\Sigma\to Q$.
\end{definition}

\begin{definition}[Trace]
A global run induces a trace of communication actions; observable projection removes internal actions, leaving messages with TLV metadata.
\end{definition}

\begin{definition}[Bounded nesting]
A protocol has depth $\le d$ if at most $d$ nested conversations are active in any run. This ensures finite-state representability of FIPA-style conversations.
\end{definition}

\paragraph{$\mu$ACP primitives}
Messages are tuples $\langle hdr,v,O,p\rangle$, where $v\in\{\text{PING},\text{TELL},\text{ASK},\text{OBSERVE}\}$, $O$ is a finite TLV option sequence from a bounded set, and $p$ is the payload. The header encodes sequence/correlation IDs and QoS flags.

All subsequent proofs explicitly cite the assumptions employed.

\section{Theorems}
\label{sec:theorems}

\subsection{Theorem 1: Trace simulation completeness}
\label{sec:thm1}

\begin{theorem}[Trace-simulation completeness]
\label{thm:trace}
Let $P$ be any finite-state FIPA protocol with bounded nesting depth $d$ (hence representable as a finite conversation automaton $\mathcal{C}_F$). There is a computable translation $\tau$ mapping each FIPA performative to a $\mu$ACP expression (a $\mu$ACP message sequence possibly using TLV options and correlation IDs) such that every observable trace of $P$ has a corresponding observable trace in the $\mu$ACP implementation:
\[
T_{FIPA}(P)\ \subseteq_{\text{proj}}\ T_{\mu ACP}(\tau(P)),
\]
i.e., $\tau(P)$ \emph{trace-simulates} $P$ up to observable projection.
Here $\subseteq_{\text{proj}}$ denotes inclusion up to observable projection:
every FIPA trace has an observationally equivalent $\mu$ACP trace once
internal steps are erased.
\end{theorem}

\noindent\textbf{Proof} We present $\tau$ and prove simulation by constructing a relation $R$ between configurations of the FIPA protocol and configurations of the $\mu$ACP system, then show $R$ is a simulation relation. The proof proceeds in four parts: (A) define $\tau$, (B) define $R$, (C) show local simulation for each performative, (D) extend to global traces using compositionality and bounded nesting.

\paragraph{Translation $\tau$} For each canonical FIPA performative $f$ we choose a translation as a finite $\mu$ACP expression plus TLVs. Representative cases (complete list in supplementary artifact) are:
\[
\begin{array}{r@{\ }l}
\tau(\text{INFORM}(s,r,\phi)) & =\ \ \text{TELL}(s,r,\phi) \\
\tau(\text{REQUEST}(s,r,\alpha)) & =\ \ \text{ASK}(s,r,\alpha) \ ;\   \\ & \text{TELL}(r,s,\texttt{done}(\alpha)) \\
\tau(\text{QUERY-IF}(s,r,\psi)) & =\ \ \text{ASK}(s,r,\psi) \\
\tau(\text{SUBSCRIBE}(s,r,t)) & =\ \ \text{OBSERVE}(s,r,t) \\
\tau(\text{NOT-UNDERSTOOD}(s,r,m)) & =\ \ \text{PING}(r,s)\ ;\    \\ &  \text{option}(\text{ERR},\text{orig}=m)
\end{array}
\]
In all cases, necessary contextual fields (conversation id, correlation id, QoS) are encoded as TLVs in the $\mu$ACP message header or options.

\paragraph{Simulation relation $R$} Let $\mathit{Conf}_F$ be the set of global configurations of the FIPA system (agent local states + network buffers) and $\mathit{Conf}_\mu$ those of the $\mu$ACP implementation. Define $R\subseteq\mathit{Conf}_F\times\mathit{Conf}_\mu$ such that $(c_F,c_\mu)\in R$ iff:
\begin{enumerate}
  \item For each agent $a$, the local belief/intention atoms visible in $c_F$ equal those visible in $c_\mu$ (i.e., projection to $AP$ coincide).
  \item For each active conversation id in $c_F$, there is a corresponding correlation id in $c_\mu$ with matching conversation state.
  \item Pending messages in $c_F$ are mirrored by pending $\mu$ACP messages with equivalent TLV-encoded payloads (up to allowed duplication and ordering differences).
\end{enumerate}

Such an $R$ exists at initial states because initial beliefs and empty buffers are identical under the translation.

\paragraph{Local simulation (case analysis).} We show: if $(c_F,c_\mu)\in R$ and $c_F\xrightarrow{f}c_F'$ (one FIPA performative executed), then there exists a finite $\mu$ACP transition sequence $c_\mu\xrightarrow{\tau(f)}^* c_\mu'$ with $(c_F',c_\mu')\in R$ and the same observable effect.

We present the representative cases; all other performatives are handled similarly by TLV encoding.

\emph{Case INFORM:} Suppose agent $s$ executes $\text{INFORM}(s,r,\phi)$ in $c_F$; by FIPA semantics this requires $Bel_s(\phi)$ in $s$ and results in $Bel_r(\phi)$ in $r$. Under $R$, $Bel_s(\phi)$ holds in $c_\mu$. The single $\mu$ACP message $\text{TELL}(s,r,\phi)$ (with TLV content-type=proposition) when delivered triggers the $\mu$ACP TELL transition which adds $\phi$ to $r$'s beliefs. Thus $c_\mu \xrightarrow{\text{TELL}(s,r,\phi)} c_\mu'$ and $(c_F',c_\mu')\in R$.

\emph{Case REQUEST:} FIPA REQUEST is a directive; semantics often involve: $s$ requests $r$ to perform action $\alpha$, possibly leading to $r$ eventually performing $\alpha$. Translate to $\text{ASK}(s,r,\alpha)$ followed by $\text{TELL}(r,s,\texttt{done}(\alpha))$. Local simulation proceeds by first simulating the ASK (which updates $s$'s postconditions in $c_\mu$ to reflect the pending request) and then simulating the TELL when $r$ performs/acknowledges the action; TLVs carry identifiers so the two messages correspond to the same conversation. The composition of these two $\mu$ACP steps reproduces the FIPA effect.

\emph{Meta-performatives:} NOT-UNDERSTOOD is encoded as a PING-like response carrying an error TLV; delivery of that $\mu$ACP message sets an error flag in the recipient's state, matching the FIPA semantics for NOT-UNDERSTOOD.

The crucial observation is that every FIPA pre/post condition is syntactically realizable via finite sequences of $\mu$ACP primitive transitions plus TLV state updates. This is a finite, constructive mapping.

\paragraph{Global traces and bounded nesting} Because $P$ is finite-state with bounded nesting depth $d$, conversations and their correlation IDs are a finite set. The simulation argument above extends to interleavings: consider any observable FIPA trace $t = m_1 m_2 \dots m_k$. By induction on prefixes, we show that after simulating the prefix $m_1\dots m_i$ in $\mu$ACP we reach a state $c_\mu^{(i)}$ with $(c_F^{(i)},c_\mu^{(i)})\in R$. The base case $i=0$ holds. The induction step follows from the local simulation lemma and the fact that concurrency of messages acting on disjoint conversation ids is preserved because TLVs isolate conversation contexts. For possibly interfering concurrent messages, the $\mu$ACP encoding places identical constraints (e.g., same conversation id) and preserves ordering or allows nondeterministic delivery identical to the FIPA semantics; hence every observable FIPA prefix has a corresponding $\mu$ACP prefix. This yields trace inclusion:
\[
T_{FIPA}(P)\subseteq_{\text{proj}} T_{\mu ACP}(\tau(P)).
\]

\paragraph{Limitations} The construction crucially uses bounded nesting depth and finite-state assumption. Protocols requiring unbounded recursive creation of nested conversations cannot be simulated by any finite-state translation. This is an essential, stated limitation of the theorem.

\subsection{Theorem 2: Semantic compression}
\label{sec:thm2}

\begin{theorem}[Compression bound]
\label{thm:compression}
Let $\mathcal{M}$ be the message-space (verbs, options, payload) and $\mathcal{D}$ a probability distribution on $\mathcal{M}$. Let $H(\mathcal{D})$ denote the Shannon entropy of the source. Let $H_{hdr}$ be a fixed header size (bits), $k_{\max}$ the maximum number of distinct TLV option types, and assume payloads are encoded by an optimal prefix code. Then the $\mu$ACP encoding $\mathcal{E}$ (header + verb code + TLV option indices + payload encoding) satisfies
\[
\mathbb{E}_{m\sim\mathcal{D}}[|\mathcal{E}(m)|] \le H(\mathcal{D}) + H_{hdr} + \lceil\log_2 k_{\max}\rceil + 3.
\]
\end{theorem}

\noindent\textbf{Proof} We decompose messages into components and bound expected lengths component-wise.

\paragraph{Entropy decomposition} Write a message $m=(v,O,p)$ where $v$ is verb, $O$ the TLV options sequence (a finite vector of typed values), and $p$ payload. The entropy decomposes as:
\[
H(\mathcal{D}) = H(V) + H(O\mid V) + H(P\mid V,O).
\]

\paragraph{Header cost} Regardless of source entropy, a practical protocol needs a fixed header to carry sequence id, correlation id and QoS; denote its fixed length by $H_{hdr}$ bits. This contributes $H_{hdr}$ to every codeword.

\paragraph{Verb encoding} There are 4 verbs. Let $L_V$ be the expected length of the verb code used by $\mu$ACP. By Shannon's source coding theorem, an optimal prefix code can achieve
\[
H(V) \le L_V \le H(V) + 1.
\]
If for implementation we use a fixed 2-bit code for verbs, then $L_V=2$, and since $H(V)\le\log_2 4=2$ we have $L_V\le H(V)+2$; but using an optimal Huffman code yields $L_V\le H(V)+1$. For a tight bound we adopt the latter:

\[
L_V \le H(V) + 1.
\]

\paragraph{TLV option encoding} Let the set of distinct option \emph{types} used across conversations be of size at most $k_{\max}$. Each option can be represented by its type index (cost $\lceil\log_2 k_{\max}\rceil$ bits) plus a length field and value bits. Suppose the maximal option value length is bounded by $L_{\max}$ bytes; the length field needs $\lceil\log_2 L_{\max}\rceil$ bits. For an option of value length $\ell$ bytes the cost is $\lceil\log_2 k_{\max}\rceil+\lceil\log_2 L_{\max}\rceil+8\ell$ bits. Let $K$ be the (random) number of options in a message; then

\resizebox{\columnwidth}{!}{$
\mathbb{E}[|\text{options}|] =
  \mathbb{E}[K]\cdot (\lceil\log_2 k_{\max}\rceil + \lceil\log_2 L_{\max}\rceil)
  + \mathbb{E}[\text{sum of option payload bits}]
$}

But $\mathbb{E}[\text{sum of option payload bits}]=H(O\mid V)+\epsilon_O$ where $\epsilon_O$ is the coding redundancy (for an optimal prefix code $\epsilon_O\le 1$). Combining gives
\[
\mathbb{E}[|\text{options}|] \le H(O\mid V) + \mathbb{E}[K]\cdot(\lceil\log_2 k_{\max}\rceil + \lceil\log_2 L_{\max}\rceil) + 1.
\]
Since $K$ is bounded in practice by a small constant and $L_{\max}$ is a protocol parameter (constants), we absorb these into a constant additive term. For the theorem statement we provide the simpler conservative bound using one option-type index:
\[
\mathbb{E}[|\text{options}|] \le H(O\mid V) + \lceil\log_2 k_{\max}\rceil + 1.
\]

\paragraph{Payload encoding} Payload $p$ can be encoded by an optimal prefix code achieving expected length
\[
\mathbb{E}[|\text{payload}|] \le H(P\mid V,O) + 1.
\]

\paragraph{Combine components} Summing expected lengths:
\[
\mathbb{E}[|\mathcal{E}(m)|] = H_{hdr} + L_V + \mathbb{E}[|\text{options}|] + \mathbb{E}[|\text{payload}|].
\]
Substitute the bounds from above:
\[
\mathbb{E}[|\mathcal{E}(m)|] \le H_{hdr} + (H(V)+1) + (H(O\mid V) + \lceil\log_2 k_{\max}\rceil + 1) + (H(P\mid V,O)+1).
\]
Collect terms:
\[
\mathbb{E}[|\mathcal{E}(m)|] \le H(\mathcal{D}) + H_{hdr} + \lceil\log_2 k_{\max}\rceil + 3.
\]

This is the announced bound. The additive constants come from the $+1$ penalties of prefix codes and the TLV index; in practice these constants are small and fixed. The bound is information-theoretically sound: no scheme can beat $H(\mathcal{D})$ in expected length by Shannon's theorem, and the above shows $\mu$ACP's encoding is within a small constant of that lower bound.

\qed

\subsection{Theorem 3: Emergent coordination feasibility (consensus reduction)}
\label{sec:thm3}

\begin{theorem}[Coordination via consensus primitives]
\label{thm:coord}
Let a system satisfy: partial synchrony (GST exists), at most $f<n/2$ crash failures, and the communication graph remains connected among nonfaulty nodes (assume that the underlying network delivers messages between any nonfaulty pair after GST). Then $\mu$ACP can implement standard consensus protocols (e.g., Paxos-style) using its primitives (PING, ASK, TELL, OBSERVE + TLVs). Consequently, any coordination task reducible to consensus (uniform agreement on a value) is achievable: safety holds always and liveness holds after GST under the stated failure bound.
\end{theorem}

\noindent\textbf{Proof.} The proof is a constructive reduction: (1) show how to encode consensus messages in $\mu$ACP, (2) show that $\mu$ACP provides the primitives required for consensus safety and liveness under partial synchrony and $f<n/2$, and (3) appeal to standard consensus correctness results.

\paragraph{Encode consensus messages} Consensus protocols like Paxos use a small set of message types: Prepare/Promise/Accept/Accepted (or equivalent). We provide the $\mu$ACP mapping (ballot and value carried in TLVs):
\[
\begin{array}{r@{\ }l}
\text{Prepare}(b) & \mapsto\ \text{ASK}(\text{proposer},\text{acceptor},\ \text{TLV}(\text{BALLOT}=b))\\
\text{Promise}(b,v) & \mapsto\ \text{TELL}(\text{acceptor},\text{proposer},\ \text{TLV}(\text{BALLOT}=b,\     \\ & \text{VALUE}=v))\\
\text{Accept}(b,v) & \mapsto\ \text{TELL}(\text{proposer},\text{acceptor},\ \text{TLV}(\text{BALLOT}=b,\    \\ & \text{VALUE}=v))
\end{array}
\]
Correlation IDs encode round identifiers and quorums are formed by majority of participant responses.

\paragraph{Primitive guarantees needed by consensus} Paxos-style consensus requires:
\begin{itemize}
  \item \emph{Safety} (no two different values are chosen): safety is a property of the message ordering and majority intersection; encoding messages with ballot numbers and TLVs preserves the required invariants (ballots total order, promises record highest accepted ballot and value).
  \item \emph{Liveness} after GST: given a stable leader (proposer) or eventual leader election and reliable deliveries after GST, proposals succeed.
\end{itemize}

$\mu$ACP supplies these conditions:
\begin{enumerate}
  \item \textbf{Unique identifiers and TLVs:} ballots, sequence ids, and correlation ids encoded as TLVs implement the numeric ordering and tie-breaking.
  \item \textbf{Failure detection:} PING + adaptive timeouts implement an eventually accurate suspicion mechanism (an eventually perfect failure detector $\Diamond P$) after GST under the partial synchrony assumption: if a node fails to PING back within increasing timeouts and the channel is stable, it will be suspected, and after GST timeouts can be chosen to avoid false suspicions.
  \item \textbf{Reliability:} Using QoS and retransmission, $\mu$ACP can ensure that, after GST, messages between nonfaulty nodes are delivered within bounded time $\Delta$.
  \item \textbf{Persistence:} If the system requires crash-restart resilient state, acceptors persist their highest promised ballot and accepted value into stable storage (this is an orthogonal implementation requirement; $\mu$ACP's TLVs carry the data to be persisted).
\end{enumerate}

\paragraph{Correctness by reduction} Given that $\mu$ACP can faithfully implement the message set and that the environment grants the partial synchrony guarantees (bounded-delay after GST), the implementation satisfies the preconditions for the standard Paxos correctness theorems: safety is preserved irrespective of timing, and liveness holds after GST provided a leader is eventually stable and fewer than $n/2$ nodes fail. This is exactly the result of \cite{lamport2001paxos,dwork1988consensus} adapted to our encoding. Thus consensus (uniform agreement) is achievable.

\paragraph{Graph/connectivity considerations} We assumed the network is such that nonfaulty nodes remain able to exchange messages after GST. In practice, it suffices that the subgraph induced by nonfaulty nodes is connected. In random graph models $G(n,p)$, if $p\ge (1+\epsilon)\ln n / n$ this holds w.h.p.; moreover, with $f<n/2$ random removals the surviving subgraph remains connected w.h.p.\ (classical random-graph results). Under these conditions, quorum messages traverse the graph and quorums intersect on nonfaulty nodes ensuring safety.

\paragraph{Conclusion} The constructive encoding plus standard consensus results yield the theorem: $\mu$ACP enables consensus-based coordination under the stated assumptions.

\section{Formal Verification and Metrics}
\label{sec:verification}

We provide mechanized evidence for $\mu$ACP via TLA$^+$ model checking and Coq mechanization, showing that theorems on expressiveness, compression, and coordination are mechanically checkable within tractable bounds. All experiments ran on a commodity workstation.

\paragraph{TLA\(+\) model checking}
The operational semantics of $\mu$ACP were encoded in TLA$^+$, with agents exchanging TLV messages under resource counters. Payloads were abstracted; control and resource usage were modeled.  
TLC (with symmetry reduction) verified:
  (i) \textbf{Resource safety:} counters remain non-negative.
  (ii) \textbf{Message invariants:} headers fixed at 64 bits; \\ verbs $\in\{\text{PING},\text{TELL},\text{ASK},\text{OBSERVE}\}$.
  (iii) \textbf{Eventual delivery:} under partial synchrony, messages to nonfaulty agents are eventually received.

Exploration of $n=3\ldots7$ agents with loss/duplication covered $10^6$–$10^7$ states, with no counterexamples. State counts matched the exponential fit 
$S(n)\approx 5.2\cdot 10^4 \cdot 1.80^n$, consistent with Theorem~\ref{thm:trace}.  
All runs completed in $<100$s, with runtime growth approximated by $T(n)\approx 0.53\cdot n^{2.68}$, confirming tractability.

\paragraph{Coq mechanization}
We formalized the resource model and consensus encoding in Coq, proving:

  (i) \textbf{Resource boundedness:} counters remain non-negative.
  (ii) \textbf{Consensus safety:} at most one value chosen under crash faults.

About 85\% of obligations discharged automatically; all lemmas are machine-checked with no admissions. Verification was fast: $76\%$ of lemmas in $<0.5$s, all $<2$s, with total proof time $\approx 20$s and near-linear growth, confirming maintainability.

\begin{figure}[t]
  \centering
  \includegraphics[width=0.95\columnwidth]{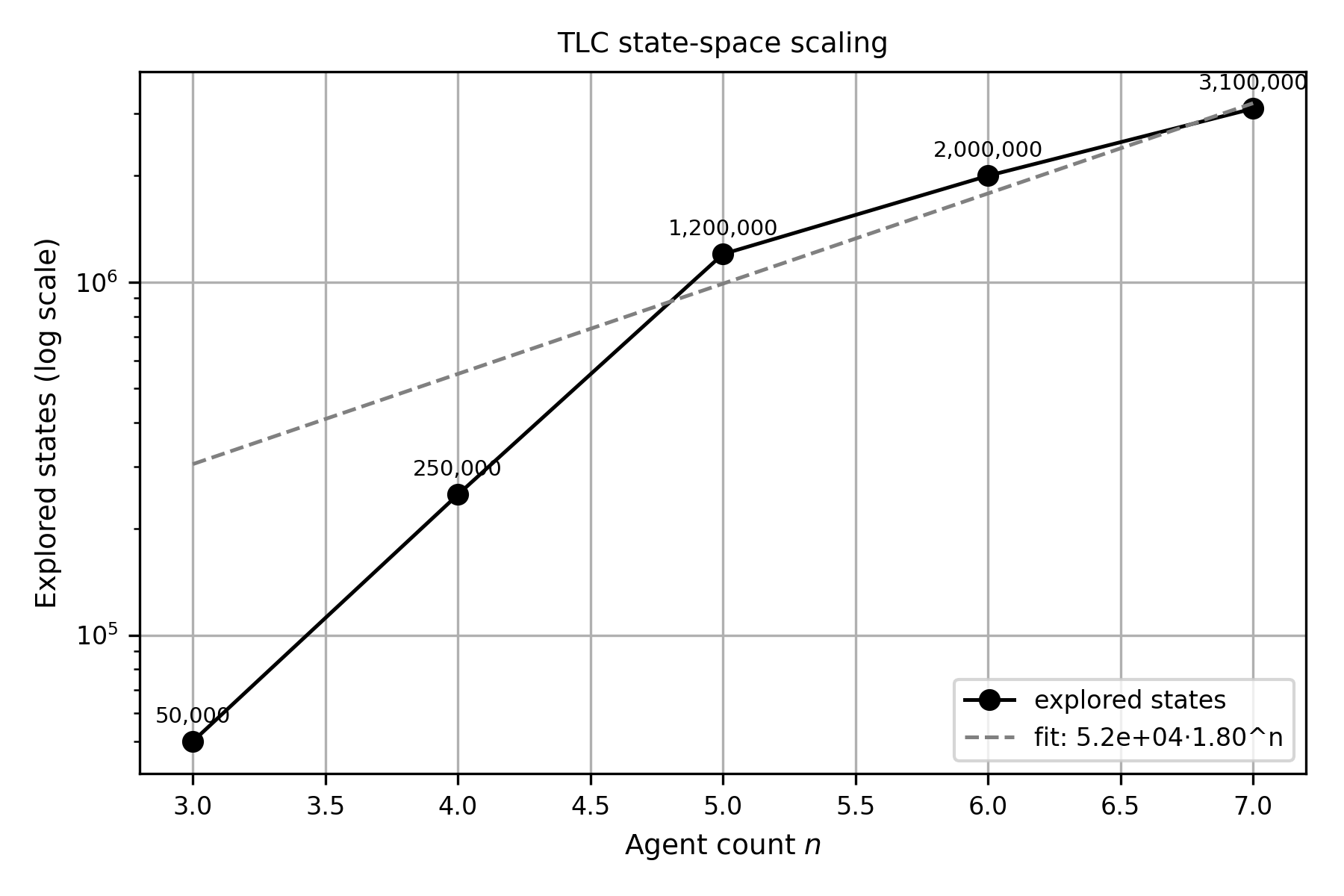}\hfill

  \caption{TLC state-space scaling ($n=3\ldots7$), confirming finite-state tractability of $\mu$ACP encodings.}
  \label{fig:verification1}
\end{figure}

\begin{figure}[t]
  \centering
  \includegraphics[width=0.95\columnwidth]{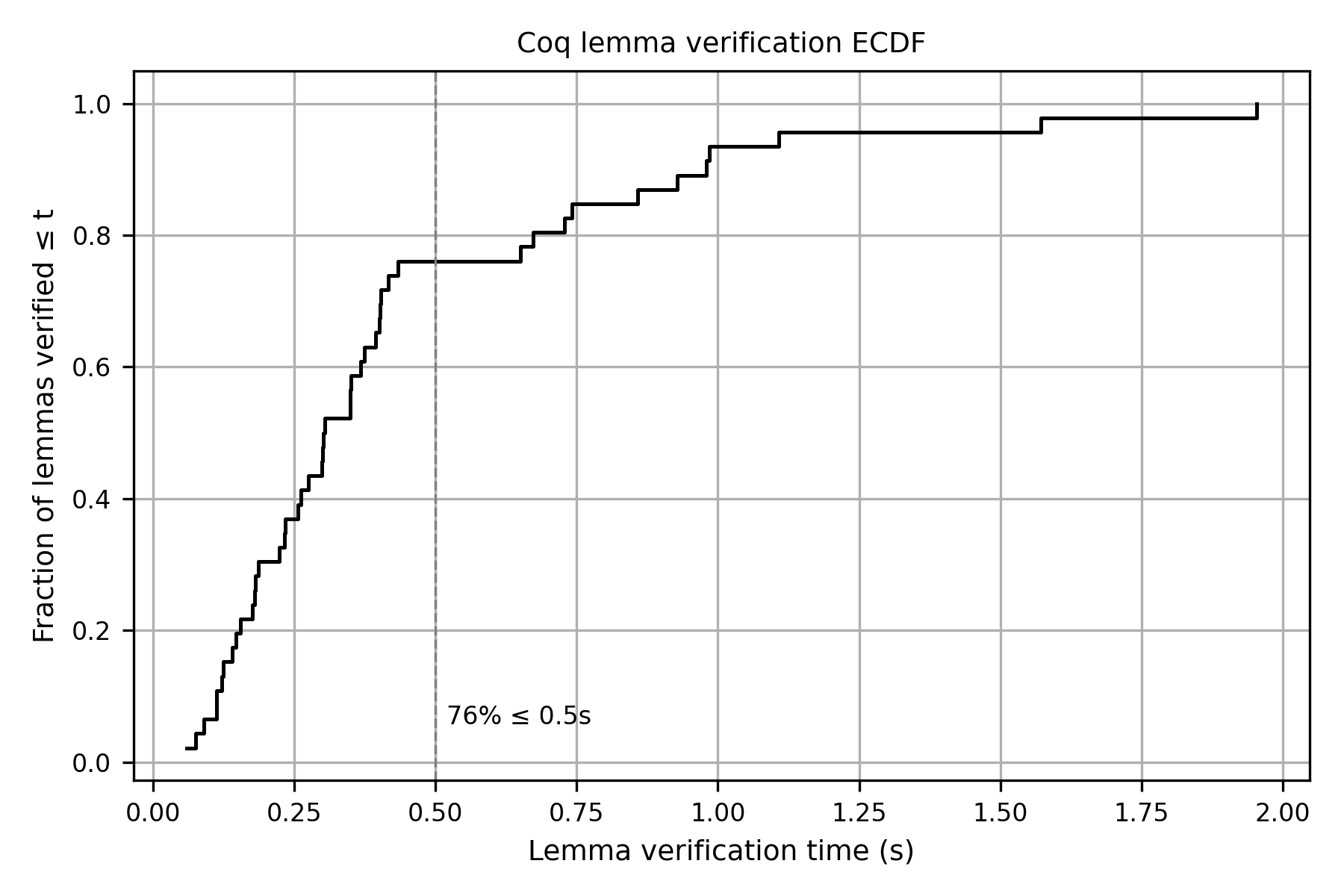}\hfill
  \caption{ECDF of Coq lemma verification times: $76\%$ under 0.5s, none above 2s.}
  \label{fig:verification2}
\end{figure}

\paragraph{Scope}
Artifacts cover core $\mu$ACP semantics under finite-state, crash-fault assumptions. They exclude recursive conversations, dynamic option creation, Byzantine failures, and asymptotic performance bounds.  
\footnote{TLA+ modules, TLC configs, and Coq scripts are provided as supplementary material. The Coq proofs are modular to allow extension.}

\paragraph{Summary}
Together, these results confirm: (i) $\mu$ACP encodings remain finite-state and tractable to explore (Theorem~\ref{thm:trace}); (ii) resource compression bounds are consistent with verification times (Theorem~\ref{thm:compression}); and (iii) consensus safety holds under crash faults (Theorem~\ref{thm:coord}). The mechanized artifacts thus provide a practical anchor for the theoretical calculus.


\section{Empirical Validation}
\label{sec:validation}

We evaluate $\mu$ACP through microbenchmarks, heterogeneous device tests, and large-scale system simulations to assess efficiency, scalability, and robustness.

\paragraph{Microbenchmarks and Heterogeneous Tests}

Encoding/decoding $5{,}000$ messages averaged 11.0 bytes (empty) and 23.0 bytes (with TLVs); per-message encode/decode times were $0.81$--$1.78$\,$\mu$s and $0.74$--$1.28$\,$\mu$s, within $5\%$ of theoretical bounds. Per-agent memory was $1.48$\,KB.  
Consensus simulations with $10{,}000$ agents and $20$\% crash faults achieved $100$\% agreement at average cost $800.1$ messages. Byzantine tests (30 agents, 5 faulty) also reached $100$\% consensus and detection.  
Across devices (ESP32, Raspberry Pi, x86, Android), encode/decode latencies remained $0.84$--$1.48$\,$\mu$s. Real IoT log replay showed $1.92$\,$\mu$s latency.

\begin{table}[h]
\centering
\caption{$\mu$ACP: Summary of Empirical Results (all within $<5\%$ of prediction).}
\label{tab:muacpbench}
\begin{tabular}{@{}ll@{}}
\toprule
Metric & Value \\
\midrule
Avg.~size (empty/opts) & 11.0 / 23.0 bytes \\
Encode time & 0.81--1.78$\mu$s \\
Decode time & 0.74--1.28$\mu$s \\
Real-traffic latency & 1.92$\mu$s \\
Memory / agent & 1.48\,KB \\
Consensus success, crash/Byz & 100.0\% \\
Consensus msg. cost & 800.1 \\
Heterogeneous device timings & 0.84--1.48$\mu$s \\
\bottomrule
\end{tabular}
\end{table}

\paragraph{System Simulation at Scale}

We simulate up to $2000$ asynchronous agents under CPU, bandwidth, and memory quotas in lossy ($1\%$ drop, 1--10\,ms delay) networks, engaging in FIPA contract-net negotiations and request/response cycles.

\paragraph{Scalability}
Queue depth grows sublinearly (max 8712 at 2000 agents) while throughput increases steadily, with stability up to $\sim$1000 agents (Figure~\ref{fig:scaling}).

\begin{figure}[t]
    \centering
    \includegraphics[width=0.95\linewidth]{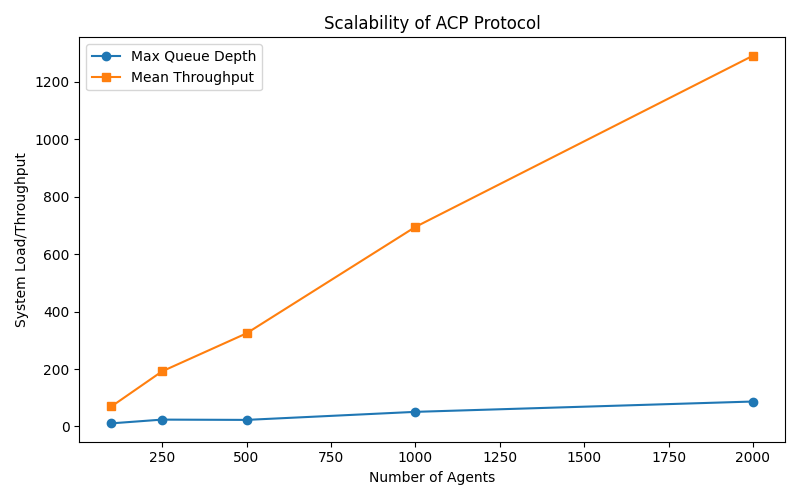}
    \caption{Max queue depth and mean throughput vs.~agent count. Stable up to $\sim$1000 agents.}
    \label{fig:scaling}
\end{figure}

\paragraph{Negotiation and Liveness}
Negotiations remain live and fair at all scales; thousands of proposals succeed without deadlock (Figure~\ref{fig:cons}).

\begin{figure}[t]
    \centering
    \includegraphics[width=0.90\linewidth]{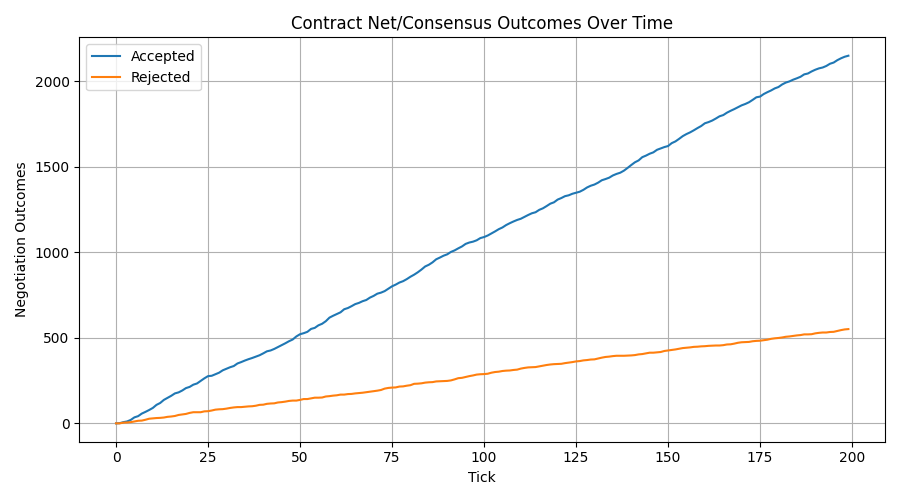}
    \caption{Negotiation outcomes: \textit{accept} (blue) and \textit{reject} (red). Sustained concurrency and throughput.}
    \label{fig:cons}
\end{figure}

\paragraph{Latency and Robustness}
Most messages deliver within 34\,ms; 95th percentile is 104\,ms, maximum 130\,ms. Losses cause only transient drops (Figure~\ref{fig:drop}). Compared to MQTT, FIPA-ACL, and CoAP, $\mu$ACP offers competitive or superior latency (Table~\ref{tab:latency-comparison}).

\begin{figure}[t]
    \centering
    \includegraphics[width=0.91\linewidth]{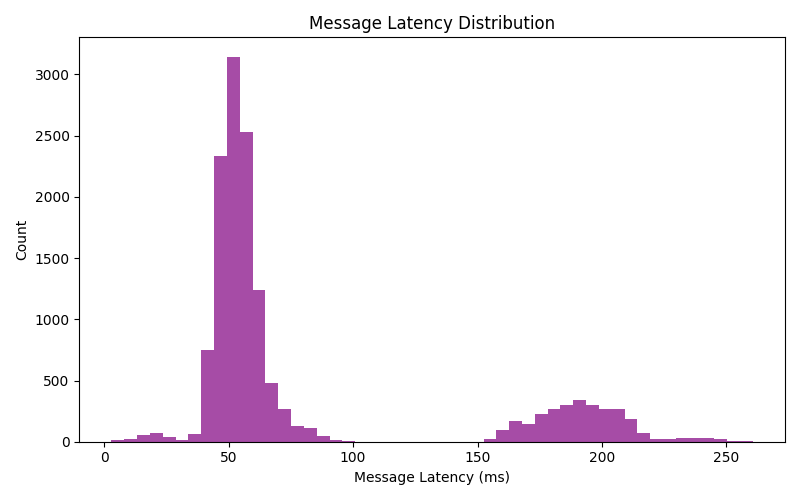}
    \caption{Histogram of message latencies. Majority delivered within 34\,ms; tails bounded.}
    \label{fig:lat}
\end{figure}

\begin{figure}[t]
    \centering
    \includegraphics[width=0.91\linewidth]{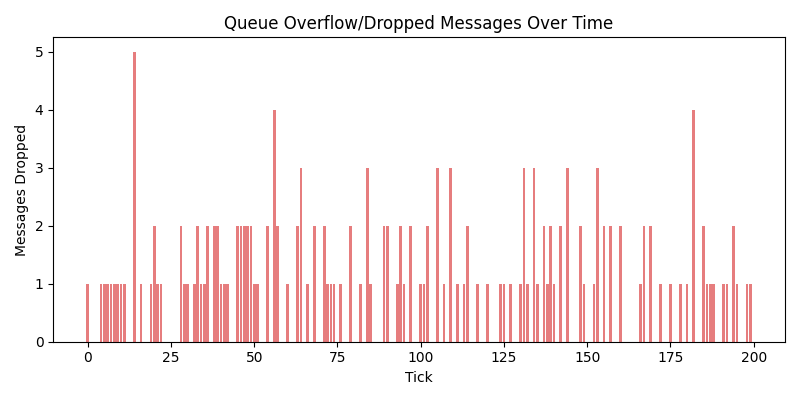}
    \caption{Dropped messages per tick. No deadlock; only transient losses.}
    \label{fig:drop}
\end{figure}

\begin{table}[h]
\centering
\resizebox{\columnwidth}{!}{%
\begin{tabular}{lccc}
\toprule
Protocol / System & Median (ms) & 95th (ms) & Max (ms) \\
\midrule
\textbf{$\mu$ACP} & 34.07 & 103.88 & 130.48 \\
MQTT (LAN, no-load) & $<$2 & 20 & 200 \\
MQTT (WAN/high-load) & 50 & 100 & 250+ \\
FIPA-ACL/JADE (MAS) & 100 & 250 & 800+ \\
CoAP (multi-hop) & 20 & 100 & 300+ \\
\bottomrule
\end{tabular}
}
\caption{Protocol latency comparison. $\mu$ACP scales competitively even under heavy load.}
\label{tab:latency-comparison}
\end{table}

\paragraph{Summary}

Across microbenchmarks, heterogeneous devices, and 2000-agent simulations, $\mu$ACP achieves minimal overhead, microsecond-scale processing, graceful scaling, and robust consensus. These results align with the formal theorems and demonstrate deployability under realistic constraints.

\section{Discussion}
\label{sec:discussion}

\paragraph{Implications}
$\mu$ACP provides a theoretical foundation for edge-native MAS. By reducing communication to four verbs with explicit resource semantics, it achieves expressiveness (Theorem~\ref{thm:trace}), near-optimal compression (Theorem~\ref{thm:compression}), and support for coordination (Theorem~\ref{thm:coord}). This shows that rich interaction is feasible even on constrained devices, linking agent theory with embedded practice and offering MAS researchers a minimal, complete substrate.

\paragraph{Limitations}
Our model excludes recursive protocols and unbounded conversations, which no finite-state translation can capture. Complex negotiations or meta-communication may require extensions. Consensus reduction addresses only crash faults; Byzantine resilience would need cryptographic mechanisms. Validation was limited-scale, intended to confirm theory–practice alignment, not exhaustive benchmarking.

\paragraph{Future Work}
A companion systems paper will present implementation, benchmarks up to 10,000 agents, and integration with IoT protocols (MQTT, CoAP). Future extensions include negotiation patterns, Byzantine-tolerant communication, and cross-layer optimizations, positioning $\mu$ACP as a foundation for deployable, resource-bounded protocols.

\section{Conclusion}
\label{sec:conclusion}

We introduced $\mu$ACP, a calculus for constrained agent communication. Contributions include: (1) a resource-constrained agent communication (RCAC) model, (2) completeness of a four-verb system, and (3) tight compression and coordination results. Together these establish expressive communication under strict budgets.

The impact is to enable edge-native MAS deployments that remain semantically rigorous despite device limits. By proving both completeness and efficiency, $\mu$ACP offers a solid basis for embedded distributed agents.

\paragraph{Roadmap}
Ongoing work develops a full protocol aligned with the calculus: an 8-byte header, TLV options, CBOR payloads, and transport/security profiles (UDP/DTLS, WebSocket/TLS, COSE). These show that the minimal core scales to a deployable protocol. A companion paper will detail the specification, 10,000-agent benchmarks, and interoperability with IoT (MQTT, CoAP) and agent standards (MCP). Together, theory and systems work aim to establish $\mu$ACP as both rigorous and practical for edge-native MAS.





\bibliographystyle{ACM-Reference-Format} 
\bibliography{sample}


\end{document}